%% file: main.tex
\documentclass[runningheads]{llncs}
\usepackage{amsmath}
\usepackage{amssymb}
\usepackage{xspace}
\usepackage{xcolor}
\usepackage{complexity}
\usepackage{url}
\usepackage{graphicx}
\usepackage{caption}
\usepackage{subcaption}
\usepackage{multirow}
\usepackage{hyperref}

\usepackage{amsmath}

\usepackage{amsthm}

  \usepackage{tikz}
  \usetikzlibrary{positioning,fit,calc,shapes,backgrounds}
  \usetikzlibrary{matrix}

\newcommand{\N}{{\mathbb{N}-\{0\}}}
\newcommand{\Z}{{\mathbb{Z}}}

\usepackage{soul}

\title{K\"{o}nigsberg Sightseeing: Eulerian Walks in Temporal Graphs}
\author{ }\author{Andrea Marino\inst{1}\and Ana Silva\inst{2}}
\authorrunning{Andrea Marino \and Ana Silva}
\institute{Dipartimento di Sistemi, Informatica, Applicazioni, Universit\`{a} degli Studi di Firenze, Firenze, Italy
\email{andrea.marino@unifi.it}\and
Departamento de Matem\'{a}tica, Universidade Federal do Cear\'{a}, Fortaleza, CE, Brazil
\email{anasilva@mat.ufc.br}}

\begin{document}

\maketitle

\newcommand{\inc}{\textsc{temp}\xspace}
\newcommand{\strinc}{\textsc{strict temp}\xspace}
\newcommand{\pwalk}{\textsc{Eulerian Walk}\xspace}
\newcommand{\pclosedwalk}{\textsc{Eulerian Closed Walk}\xspace}
\newcommand{\plocaltrail}{\textsc{Eulerian Local Trail}\xspace}
\newcommand{\plocaltour}{\textsc{Eulerian Local Tour}\xspace}
\newcommand{\ptrail}{\textsc{Eulerian Trail}\xspace}
\newcommand{\ptour}{\textsc{Eulerian Tour}\xspace}
\newcommand{\walk}{Eulerian walk\xspace}
\newcommand{\localtrail}{Eulerian local trail\xspace}
\newcommand{\localtour}{Eulerian local tour\xspace}
\newcommand{\trail}{Eulerian trail\xspace}
\newcommand{\tour}{Eulerian tour\xspace}

\begin{abstract}

An Eulerian walk (or Eulerian trail) is a walk (resp. trail) that visits every edge of a graph $G$ at least (resp. exactly) once. This notion was first discussed by Leonhard Euler while solving the famous Seven Bridges of K\"{o}nigsberg problem in 1736. What if Euler had to take a bus? In a temporal graph $(G,\lambda)$, with $\lambda: E(G)\to 2^{[\tau]}$, an edge $e\in E(G)$ is available only at the times specified by $\lambda(e)\subseteq [\tau]$, in the same way the connections of the public transportation network of a city or of sightseeing tours are available only at scheduled times. In this scenario, even though several translations of Eulerian trails and walks are possible in temporal terms, only a very particular variation has been exploited in the literature, specifically for infinite dynamic networks (Orlin, 1984). In this paper, we deal with temporal walks, local trails, and trails, respectively referring to edge traversal with no constraints, constrained to not repeating the same edge in a single timestamp, and constrained to never repeating the same edge throughout the entire traversal. We show that, if the edges are always available, then deciding whether $(G,\lambda)$ has a temporal walk or trail is polynomial, while deciding whether it has a local trail is $\NP$-complete even if it has lifetime~2. In contrast, in the general case, solving any of these problems is $\NP$-complete, even under very strict hypothesis.
\end{abstract}

\section{Introduction}

An Eulerian walk (or Eulerian trail) is a walk (resp. trail) that visits every edge of a graph $G$ at least (resp. exactly) once. The Eulerian trail notion was first discussed by Leonhard Euler while solving the famous Seven Bridges of K\"{o}nigsberg problem in 1736, where one wanted to pass by all the bridges over the river Preger without going twice over the same bridge. Imagine now a similar problem, where you have a set of sights linked by possible routes. If the routes themselves are also of interest, a sightseeing tourism company might want to plan visits on different days that cover all the routes. One could do that with no constraints at all (thus performing a walk), or with the very strict constraint of never repeating a route (thus getting a trail), or constraining oneself to at least not repeting the same route on the same day (thus getting what we called a local trail). If we further assume that some routes might not be always accessible, we then get distinct problems defined on temporal graphs.

In a temporal graph $(G,\lambda)$, with $\lambda: E(G)\to 2^{[\tau]}$, an edge $e\in E(G)$ is available only at the times specified by $\lambda(e)\subseteq [\tau]$, in the same way the connections of the public transportation network of a city or of sightseeing tours are available only at scheduled times.
In this scenario, paths and walks are valid only if they traverse a sequence of adjacent edges $e_1, \ldots, e_k$ such that for each $i\in [k-1]$, $\lambda(e_i)\leq \lambda(e_{i+1})$, i.e. whose time sequence is non-decreasing (similarly, strictly increasing sequences, i.e. with $\lambda(e_i)< \lambda(e_{i+1})$, can be considered).

Several translations of Eulerian trails and walks are possible in temporal terms, depending on the constraints we consider. In particular, we study the following variations. Below, all the walks and trails are implicitly  considered to be temporal, as defined in the previous paragraph.

\begin{problem} Given a temporal graph $(G,\lambda)$, we consider the following problems:
\begin{itemize}
\item \pwalk: deciding whether $(G,\lambda)$ has an Eulerian walk, i.e. a walk traversing each edge at least once.
\item \plocaltrail: deciding whether $(G,\lambda)$ has an Eulerian local trail, i.e. a walk traversing each edge at least once and at most once in each time stamp.
\item \ptrail: deciding whether $(G,\lambda)$ has an Eulerian trail, i.e. a walk traversing each edge exactly once.
\end{itemize}
\label{prob:one}
\end{problem}

We also consider the related problems where the walks/trails are closed (first vertex equal to the last one), respectively referring to them as \pclosedwalk, \plocaltour, and \ptour.

The research on temporal graphs have attracted a lot of attention in the past decade (we refer the reader to the surveys~\cite{M15,LVM.18} and the seminal paper~\cite{KKK00}), and they appear also under different names, e.g. as time-varying graphs~\cite{CFQS12}, as evolving networks~\cite{BorgnatFGMRS07}, and as link streams~\cite{LVM.18}. Even after all the received attention, surprisingly enough none of the above problems have been previously considered. 

Concerning \pwalk, one of the closest concepts is the one defined by Orlin in~\cite{orlin1984some}, where he gives a polynomial algorithm to check the existence of an Eulerian closed walk (i.e. a \emph{tour}) in dynamic graphs. However, the dynamic graph model is quite different from the temporal graph model used in this paper, as pointed out in~\cite{M15}. Indeed, looking at the corresponding time-expanded graph, temporal edges can go back in time and the graph is infinite. Nevertheless, the results presented there seemed to point towards the polinomiality of the problems investigated here, as observed in~\cite{M15}: ``the results proved for it [the dynamic graph model] are resounding and possibly give some first indications of what to expect when adding to combinatorial optimization problems a time dimension". We found however that this is not the case, as we will show that even \pwalk turns out to be much harder on temporal graphs (for more details about dynamic graphs and the result proved by Orlin in~\cite{orlin1984some} see Appendix~\ref{app:rel}). Taking inspiration in~\cite{orlin1984some}, we also define a dynamic-based temporal graph as a temporal graph whose edges are always available, and we analyze the complexity of the above problems on these particular instances. 

Still concerning \pwalk, a closely related problem is the \textsc{Temporal Exploration problem} (TEXP)~\cite{michail2016traveling}, which consists of, given a temporal graph $(G,\lambda)$, finding a temporal walk that visits all vertices in $G$ (possibly, more than once) whose arrival time is minimum. In~\cite{michail2016traveling}, they prove that this problem is \NP-complete and even not approximable unless $\P=\NP$; this is in stark contrast with the static version of the problem, which can be trivially solved in linear time. However this is quite different wrt our \pwalk, as we are considering walks passing through all the edges instead of all the vertices, and this difference is crucial as a transformation from TEXP to \pwalk does not seem to be easy. Indeed, by simply transforming each vertex into an edge, we would get two types of edges, connection-edges and vertex-edges, and a negative answer to \pwalk could be just motivated by the impossibility of visiting all the connection-edges, which are actually not required to be visited to solve TEXP. This is not surprising as we are trying to transform a Hamiltonian walk (i.e. a walk passing by each vertex) into an Eulerian walk. We mention that, unlike the problems investigated here, a lot of research has been devoted to temporal node exploration, e.g. bounding the arrival time of such walks in special instances~\cite{erlebach2018faster,Erlebach0K15} and extending previous results in the case of non-strictly increasing paths~\cite{ErlebachS20}.

Concerning \ptrail and \plocaltrail, observe that, when $\tau =1$, then both of them degenerate into the original formulation of the Seven Bridges of K\"{o}nigsberg problem. This is why we think they appear to be more natural adaptations of the static version of the problem. Nevertheless they have never been investigated before, up to our knowledge. 

\paragraph{Our results} Our results are summarized in Table~\ref{tab:results} and detailed in Theorem~\ref{thm:main}.

\begin{theorem}
\label{thm:main}
Given a temporal graph $(G,\lambda)$
\begin{enumerate}
    \item \label{item:one}\pwalk is \NP-complete, even if each snapshot of $(G,\lambda)$ is a forest of constant size, while it is is polynomial for bounded $\tau$. It is also polynomial if $(G,\lambda)$ is dynamic-based. 
    
    \item  \label{item:two}\plocaltrail is \NP-complete for each $\tau\geq 2$ and in the case $(G,\lambda)$ is dynamic-based.
    
    \item \label{item:three}\ptrail is \NP-complete for each $\tau\geq 2$. It is polynomial if $(G,\lambda)$ is dynamic-based.  
\end{enumerate}
Same applies to tours, i.e. \pclosedwalk, \plocaltour, and \ptour.
\end{theorem}

Theorem~\ref{thm:main} gives a complete taxonomy of our problems, also focusing on the possibility of getting polynomial algorithms when we have a small lifetime $\tau$. In particular, for \ptrail and \plocaltrail, since they become polynomial when $\tau=1$, the bound for $\tau$ is optimal, giving us a complete dichotomy with respect to the lifetime of $(G,\lambda)$. In contrast \pwalk is easily solvable for every fixed $\tau$, showing that walks are easier than trails even on the temporal context. 

It is important to remark that none of the above variations immediately implies any of the others. We will show indeed that the property of being Eulerian for the static base graph $G$ is in general a necessary but not sufficient condition for the existence of an \trail, becoming sufficient only if we restrict to dynamic-based temporal graphs. In the case of \localtrail, we will see that this property is not even necessary.

Finally, as a by product of our reductions we get the following result about static graphs, which can be of independent interest.
\begin{corollary}
Given a graph $G$, deciding whether all the edges of $G$ can be covered with two trails is \NP-complete.
\label{cor:cover}
\end{corollary}

\begin{table}[t]
    \centering
    \begin{tabular}{|c|c|c|c|}
    \hline
          &  \pwalk & \plocaltrail & \ptrail\\
         \hline
         \multirow{2}{*}{\textsc{General}} & \NP-c for $\tau$ unbounded & \NP-c [from below] & \NP-complete \\
          & Poly for $\tau$ fixed  & & for $\tau=2$\\
          \hline
         \textsc{Dynamic-Based} & Poly$^\dagger$ & \NP-c for $\tau=2$ & Poly$^\star$ \\         
         \hline
    \end{tabular}
    \caption{\textbf{Our results concerning Problem~\ref{prob:one}.} For general temporal graphs (first row) and for dynamic-based temporal graphs (second row). $^\dagger$ corresponds to deciding whether $G$ is connected. $^\star$ corresponds to deciding whether $G$ has an Eulerian trail. }
    \label{tab:results}
\end{table}

\paragraph{Further Related Work.}

Other than the papers on TEXP previously mentioned, there is a vast literature about finding special paths or walks in temporal graphs, and some interesting papers include~\cite{CasteigtsHMZ20,WCHLX14,SQFCA11,michail2016traveling}. Also, it is largely known that a static graph $G$ has an Eulerian tour (trail) if and only if $G$ has at most one non-trivial component and all the vertices have even degree (at most two vertices have odd degree). A graph is called Eulerian if it has an Eulerian tour.

When considering dynamic-based temporal graphs, as edges are assumed to be always available during the lifetime $\tau$, we could relate our problems to several other problems on static graphs.
A closely related one would be the Chinese Postman problem, where the edges of the graph have positive weights and one wants to find an Eulerian closed walk on $G$ with minimum weight; in other words, one wants to add copies of existing edges in order to obtain an Eulerian graph of minimum sum weight. Even if we regard the Chinese Postman problem where the weights are all equal to~1, this is very different from our approach since for us, repetition of a long common trail in different snapshots does not make the solution worse, while it would considering the Chinese Postman problem. It is easy to see though that the solution for the Chinese Postman would give us an upper bound for the amount of time spent on an Eulerian local tour of a dynamic-based graph, as we could start a new trail on a new snapshot whenever an edge repetition was detected. The Chinese Postman problem is largely known to be polynomial~\cite{gljan1962graphic}, and some variations that take time into consideration have been investigated, mostly from the practical point of view (see e.g.~\cite{ccodur2020time,sun2011dynamic,wang2002time}), but none of which is equivalent to our problem.

The problem of trying to obtain an Eulerian subgraph (as opposed to a supergraph) has also been studied. In~\cite{cygan2014parameterized}, the authors study a family of problems where the goal is to make a static graph Eulerian by a minimum number of deletions. They completely classify the parameterized complexity of various versions of the problem: vertex or edge deletions, undirected or directed graphs, with or without the requirement of connectivity. Also in~\cite{fomin2014long}, the parameterized complexity of the following Euler subgraph problems is studied: (i) Largest Euler Subgraph: for a given graph $G$ and integer parameter $k$, does $G$ contain an induced Eulerian subgraph with at least $k$ vertices?; and (ii) Longest Circuit: for a given graph $G$ and integer parameter $k$, does $G$ contain an Eulerian subgraph with at least $k$ edges? 

\plocaltrail on dynamic-based graphs is actually more closely related to the problem of covering the edges of a graph with the minimum number of (not necessarily disjoint) trails, whereas the aforementioned problems are more concerned with either minimizing edge repetitions or maximizing the subgraph covered by a single trail. Even if the trail cover problem can be so naturally defined and involve such a basic structure as trail, up to our knowledge it has not yet been previously investigated. Note that \plocaltrail is slightly different from trail cover, since we also require that together the trails form a walk. In any case, a small modification of our proof of Theorem~\ref{thm:main}.\ref{item:two} implies that deciding whether the edges of a graph can be covered with at most two trails is $\NP$-complete (Corollary~\ref{cor:cover}). Interestingly enough, the vertex version of this problem, namely the path cover problem, has been largely investigated (see e.g. \cite{arumugam2013decomposition,gomez2018covering,manuel2018revisiting}).

\paragraph{Preliminaries.}

We use and extend the notation in~\cite{M15}. A temporal graph is a graph together with a function on the edges saying when each edge is active; more formally, a \emph{temporal graph} is a pair $(G,\lambda)$, where $\lambda:E(G)\rightarrow 2^\N$. Here, we consider only finite temporal graphs, i.e., graphs such that $\max\bigcup_{e\in E(G)}\lambda(e)$ is defined. This value is called the \emph{lifetime of $(G,\lambda)$} and denoted by $\tau$. Given $i\in [\tau]$, we define the \emph{snapshot $G_i$} as being the subgraph of $G$ containing exactly the edges active in time $i$; more formally, $V(G_i)=V(G)$ and $E(G_i) = \{e\in E(G)\mid i\in \lambda(e)\}$.

Given vertices $v_0,v_k$ in a graph $G$, a \textit{$v_0,v_k$-walk} in $G$ is an alternating sequence $(v_0, e_1, v_1,  \ldots, e_k, v_k)$ of vertices and edges such that $e_i$ goes from $v_{i-1}$ to $v_{i}$ for $i\in \{1, \ldots, k\}$. We define a walk in a temporal graph similarly, except that a walk cannot go back in time. More formally, given a temporal graph $(G,\lambda)$ and a $v_0,v_k$-walk $W = (v_0, e_1, v_1,  \ldots, e_k, v_k)$, we say that $W$ is a \textit{temporal $v_0,v_k$-walk} if $\lambda(e_1)\le \lambda(e_2)\le \ldots\le\lambda(e_k)$. It is \emph{closed} if it starts and finishes on the same vertex of $G$, i.e., if $v_0=v_k$.  

We say that a temporal walk $W$ is a \emph{local trail} if there are no two occurrences of the same edge of $G$ in the same snapshot, i.e., if $W$ restricted to $G_i$ is a trail in $G$ for every $i\in [\tau]$. We say that $W$ is a \emph{trail} if there are no two occurrences of the same edge of $G$ in $W$. A closed (local) trail is also called a \emph{(local) tour}. Finally, a temporal walk $W$ is called \emph{Eulerian} if at least one copy of each edge of $G$ appears at least once in $W$. Observe that, by definition, an Eulerian trail visits every edge exactly once.

A \emph{dynamic-based graph} is a temporal graph $(G,\lambda)$ where the edges are always available.\footnote{This is the reason why we use the term dynamic-based, as they are similar to the dynamic networks used in~\cite{orlin1984some} when studying Eulerian trails, except that edges cannot go back in time and the lifetime is finite.} We denote a dynamic-based graph simply by $(G,[\tau])$ where $\tau$ is the lifetime of the temporal graph. 

\section{Eulerian Walk}

In this section we focus on \pwalk, i.e. deciding if there is a temporal walk passing by each edge at least once, proving the results in Item~\ref{item:one} in Theorem~\ref{thm:main}, summarized in the first column of Table~\ref{tab:results}.

In particular, a preliminary result, whose formal proof is reported in Appendix~\ref{sec:easy}, concerns the case where the lifetime $\tau$ is bounded. It consists basically of checking whether there is a choice of connected components $H_1, \ldots, H_\tau$, one for each timestamp $i$, that together cover all the edges of $G$ and is such that $H_i$ intersects $H_{i+1}$, for each $i\in [\tau-1]$.

\begin{lemma}
Given a temporal graph $(G,\lambda)$ with fixed lifetime $\tau$, solving \pwalk on $(G,\lambda)$ can be done in time $O((n+m)\cdot n^{\tau-1})$, where $n = |V(G)|$ and $m = |E(G)|$.
\label{lem:easy}
\end{lemma}

In the following, we show that when $\tau$ is unbounded, deciding whether $(G,\lambda)$ admits an \walk is $\NP$-complete by reducing from 3-\SAT. This is best possible because of the above lemma. 

\begin{theorem}
Given a temporal graph $(G,\lambda)$, deciding whether $(G,\lambda)$ admits an \walk is $\NP$-complete, even if each snapshot of $(G,\lambda)$ is a forest of constant size.
\label{thm:walk}
\end{theorem}
\begin{proof}
We make a reduction from 3-\SAT. Let $\phi$ be a 3-CNF formula on variables $\{x_1,\cdots, x_n\}$ and clauses $\{c_1,\cdots,c_m\}$, and construct $G$ as follows. For each clause $c_i$, add vertices $\{a_i,b_i\}$ to $G$ and edge $a_ib_i$. Now consider a variable $x_i$, and let $c_{i_1},\cdots, c_{i_p}$ be the clauses containing $x_i$ positively, and $c_{j_1},\cdots, c_{j_q}$ be the clauses containing $x_i$ negatively. Add two new vertices $x_i,\overline{x}_i$ to $G$, and edges $\{x_ia_{i_k}\mid k\in [p]\}\cup \{\overline{x}_ia_{j_k}\mid k\in [q]\}$; denote the spanning subgraph of $G$ formed by these edges by $H_i$, and let $H'_i$ be equal to $H_i$ together with edges $\{a_ib_i\mid i\in \{i_1,\cdots,i_p,j_1,\cdots, j_q\}$. We can suppose that $\{i_1,\cdots, i_p\}\cap \{j_1,\cdots, j_q\}=\emptyset$ as otherwise the clauses in the intersection would always be trivially valid; thus we get that $H_i,H'_i$ are forests. 
Finally, add a new vertex $T$ and make it adjacent to every vertex in $\{x_i,\overline{x}_i\mid i\in[n]\}$.

We now describe the snapshots of $(G,\lambda)$. See Figure~\ref{fig:walkGadget} to follow the construction. We first build 2 consecutive snapshots in $(G,\lambda)$ related to $x_i$, for each $i\in [n]$. The first one is equal to $H'_i$, and the second one contains exactly the edges $\{Tx_i,T\overline{x}_i,Tx_{i+1},T\overline{x}_{i+1}\}$ if $i<n$, and if $i=n$, then the second snapshot is equal to $G-\{a_jb_j\mid j\in[m]\}$; this can be done because this subgraph is connected. Denote by $S_i^1,S_i^2$ the first and second snapshot of $x_i$, for each $i\in [n]$. Put these snapshots consecutively in timestamps $1$ through $2n$, in the order of the indexing of the variables. For now, observe that only the last snapshot might not be a forest; this will be fixed later. We now prove that $\phi$ is a satisfiable formula if and only if $(G,\lambda)$ admits an \walk. 

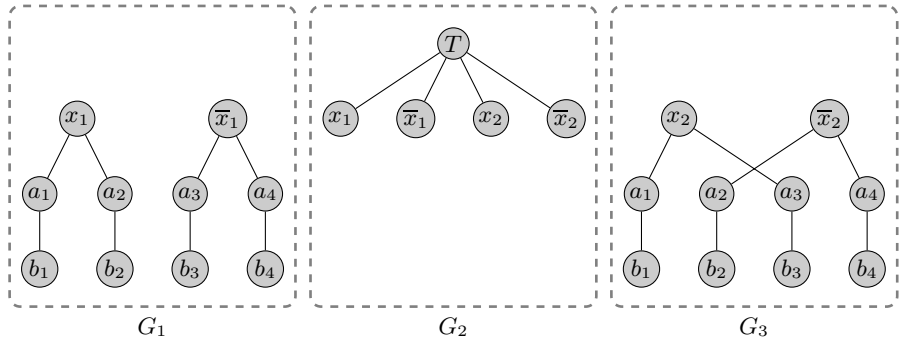
\begin{figure}[h]
\input{figs/walkGadget}
\caption{First three snapshots of the construction. In this example, we have $c_1$ containing $(x_1\vee x_2)$, $c_2$ containing $(x_1\vee \overline{x}_2)$, $c_3$ containing $(\overline{x}_1\vee {x}_2)$, and $c_4$ containing $(\overline{x}_1\vee \overline{x}_2)$.}\label{fig:walkGadget}
\end{figure}

Observe that, since we are dealing with a walk, we are allowed to repeat edges as many times as we want; hence, if we visit any vertex inside a component of a snapshot, we can also visit the entire component. At the same time, it is not possible to visit more than one component within a snapshot. This is implicitly used below.

First consider a satisfying assignment of $\phi$. Now, construct an \walk as follows. Start by visiting all the edges in the component of $H'_1$ containing $x_1$, if $x_1$ is true, or the one containing $\overline{x}_1$, otherwise. Then, in $S^2_1$, jump to $x_2$ if $x_2$ is true, or to $\overline{x}_2$, otherwise. Repeat the process until reaching $S^1_n$, and at the last snapshot, visit all the edges in $G-\{a_jb_j\mid j\in[m]\}$. Because each clause $c_i$ contains at least one true variable, we know that edge $a_ib_i$ is visited.

Now, consider an \walk $W$ of $(G,\lambda)$, and denote by $W_i^j$ the walk $W$ restricted to $S^j_i$, for each $i\in [n]$ and $j\in [2]$. We set $x_i$ to true if $W_i^1$ contains $x_i$, and to false otherwise. Now, consider a clause $c_i$ containing variables $x_{k_1},x_{k_2},x_{k_3}$. Because $a_ib_i$ appears only in snapshots $S^1_{k_1},S^1_{k_2},S^1_{k_3}$, and only in the component containing the literal that appears in $c_i$, we get that at least one of the three literals must be set to true.

Finally, observe that we could repeat the same pattern as the one in the first $2n-1$ snapshots in order to visit the remaining edges in $G-\{a_jb_j\mid j\in[m]\}$. As long as the edges in $\{a_ib_i\mid i\in[m]\}$ appear only in the first $2n-1$ snapshots, the same argument as before still applies, and we get the further constraint that each snapshot is a forest. Additionally, they can also be considered to have constant size since 3-\SAT is $\NP$-complete even if each variable appears at most three times~\cite{dahlhaus1994complexity}.
\end{proof}

Now, if we consider a dynamic-based graph $(G,\lambda)$, since all the edges are active throughout its lifetime, we clearly have that there exists an \walk if and only if $G$ is connected, as highlighted by the following Lemma.

\begin{lemma}
\pwalk is polynomial for dynamic-based temporal graphs.
\label{lem:pwalkdynamic}
\end{lemma}

By Lemma~\ref{lem:easy}, Theorem~\ref{thm:walk}, and Lemma~\ref{lem:pwalkdynamic}, we obtain Item~\ref{item:one} of Theorem~\ref{thm:main}.
Finally, note that if one is interested in closed walks instead, not only our $\NP$-completeness reduction can be adapted in order to ensure that we can always go back to the initial vertex, but also the complexity results still hold.

\section{Eulerian Local Tours and Trails}\label{sec:localtrail}

In this section we focus on Item~\ref{item:two} of Theorem~\ref{thm:main}. In the whole section, we will focus on dynamic-based temporal graphs as the hardness results for general temporal graphs are implied by the ones we prove for this restricted class. After the preliminary result in Lemma~\ref{lem:oddVertices}, we focus on proving the hardness result for the problem of deciding whether $(G,[2])$ has an \localtour, explaining the construction behind our reduction from NAE 3-\SAT, whose correctness is proved in Theorem~\ref{thm:localtour}. We also argue that, if $G$ is a cubic graph, then being Hamiltonian is a necessary but not sufficient condition for $(G,[2])$ to admit an \localtour, arguing the need of an \emph{ad hoc} reduction for our problem.
As the reduction in Theorem~\ref{thm:localtour} focuses on solving \plocaltour for $\tau=2$, in Corollary~\ref{cor:local} we extend this result to each fixed $\tau$ and to trails, thus completing the proof of Item~\ref{item:two} of Theorem~\ref{thm:main}.

The following Lemma, whose proof is given in Appendix~\ref{app:oddVertices}, will help our proof.

\begin{lemma}\label{lem:oddVertices}
Let $G$ be a graph. If $(G,[2])$ has an \localtour $T$, then $T$ restricted to timestamp $i$ must pass by all vertices of odd degree in $G$, for each $i\in [2]$.
\end{lemma}

A simple consequence of the above lemma is that, as previously said, if $G$ is cubic, then $G$ must Hamiltonian in order for $(G,[2])$ to have an \localtour. Since deciding whether a cubic graph is Hamiltonian is $\NP$-complete~\cite{garey1976planar}, this hints towards the $\NP$-completeness of the problem. However, since the other way around is not necessarily true (see e.g. the graph in Figure~\ref{fig:outerplanar}), we need an explicit reduction. Indeed, the construction in Figure~\ref{fig:outerplanar} shows us that we might need an arbitrarily large lifetime in order to be able to visit all the edges of $(G,[\tau])$ even if $G$ is a 2-connected outerplanar cubic graph (which is trivially hamiltonian). 

\begin{figure}[ht]
\input{figs/outerplanar}
\caption{Example of outerplanar graph $G$ such that $(G,[2])$ does not have an \localtour.}  
\label{fig:outerplanar}
\end{figure}
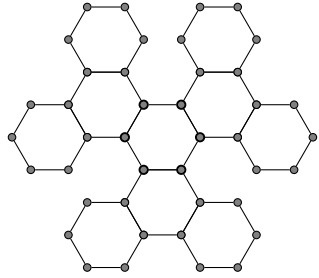

In the following we explain the construction behind our reduction from NAE 3-\SAT. Let $\phi$ be a CNF formula on variables $\{x_1,\cdots,x_n\}$ and clauses $\{c_1,\cdots,c_m\}$. 
We start by presenting a meta-construction, in the sense that part of the constructed graph will be presented for now as black boxes and the actual construction is done later, as depicted in Figure~\ref{fig:DynEdge}. The meta part concerns the clauses; so for now, denote by $C_i$ the black box related to clause $c_i$. Without going into details, $C_i$ will contain exactly one entry vertex for each of its literal. So, given a literal $\ell$ contained in $c_i$, denote by $I_i(\ell)$ the entry vertex for $\ell$ in $C_i$. All defined three vertices are distinct.

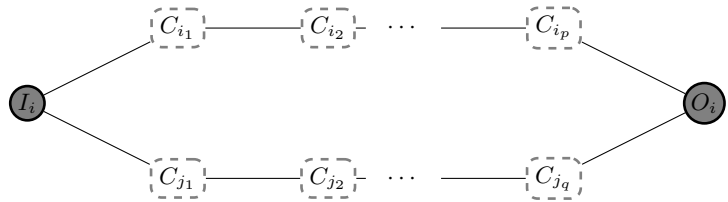
\begin{figure}[h]
\input{figs/DynEdge}
\caption{Edge gadget with clause black boxes.}\label{fig:DynEdge}
\end{figure}

Now, for each variable $x_i$, let $c_{i_1},\cdots,c_{i_p}$ be the clauses containing $x_i$ positively and $c_{j_1},\cdots,c_{j_q}$ containing $x_i$ negatively. Add two new vertices, $I_i$ and $O_i$ (these will be the entry and exit vertices for the variable gadget), and add the following edges (these compose the paths shown in Figure~\ref{fig:DynEdge}):
\[\begin{array}{ll}
E_i = & \{I_iI_{i_1}(x_i),I_iI_{j_1}(\overline{x}_i),I_{i_p}(x_i)O_i,I_{j_q}(\overline{x}_i)O_i\} \\ & \cup \{I_{i_h}(x_i)I_{i_{h+1}}(x_i)\mid h\in [p-1]\} \\ & \cup \{I_{j_h}(\overline{x}_i)I_{j_{h+1}}(\overline{x}_i)\mid h\in [q-1]\}\end{array}\]

The paths will function as a switch, telling us whether the variable is true or false within the considered snapshot; we then denote by $P_i$ the set of edges in the path $(I_i,I_{i_1}(x_i),\cdots,I_{i_p}(x_i),O_i)$, and by $\overline{P}_i$ the set of edges in the path $(I_i,I_{j_1}(\overline{x}_i),\cdots,I_{j_q}(\overline{x}_i),O_i)$. 
Now, to link the variable gadgets and to construct the clause gadgets, we will need a gadget that will function as an edge that must appear in the trail performed in $G_1$ and the one performed in $G_2$. For this, we use Lemma~\ref{lem:oddVertices} applied to the gadget in Figure~\ref{fig:forcedEdge}; when adding such a gadget between a pair $u,v$, we simply say that we are adding the \emph{forced edge $uv$}.

Now, to link the variable gadgets, we add three new vertices $s_1,s_2,t$ and the following forced edges. 

\[E' = \{s_it\mid i\in [2]\}\cup \{tI_1,O_nt\}\cup \{O_iI_{i+1}\mid i\in [n]\}.\]

The new vertices simply help us assume where the trail starts and finishes. Now, let $T$ be an \localtour of $(G,[2])$ and denote by $T_i$ the trail in $G$ defined by $T$ restricted to $G_i$, for $i\in [2]$. It is fairly easy to see (and we will prove it shortly) that if we can ensure that $T_1$ uses $P_i$ if and only if $T_2$ uses $\overline{P}_i$, then we can prove equivalence with NAE 3-\SAT. In other words, the clause gadget must be so that, for every clause $c_j$ containing $x_i$ (or equivalently $\overline{x}_i$), we get that either both edges incident to $I_j(x_i)$ in $P_i$ (or equivalently $I_j(\overline{x}_i)$ in $\overline{P}_i$) are used, or none of them is used. Such a gadget is presented in Figure~\ref{fig:DynClauseComb}, where the red edges are forced.

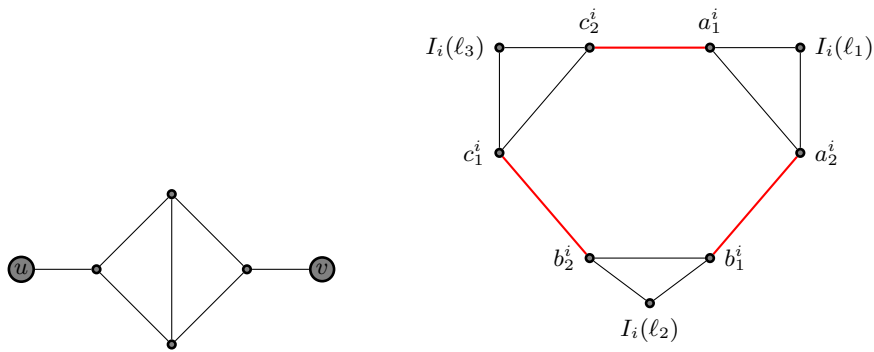
\begin{figure}[h]
     \centering
     \begin{subfigure}[b]{0.3\textwidth}
         \centering
     \input{figs/forcedEdge}
\caption{Gadget related to a forced edge $uv$.}\label{fig:forcedEdge}
     \end{subfigure}
     \hfill
     \begin{subfigure}[b]{0.6\textwidth}
         \centering
     \input{figs/DynClause}
    \caption{Gadget related to clause $c_i$. Red edges represent forced edges.}
    \label{fig:DynClauseComb}
    \end{subfigure}
    \caption{Gadgets for the reduction in Theorem~\ref{thm:localtour}.}
\end{figure}

\begin{theorem}\label{thm:localtour}
Let $G$ be a graph, with degree is at most 4. Then \plocaltour is $\NP$-complete on $(G,[2])$.
\end{theorem}
\begin{proof}
Let $\phi$ and $G$ be as previously stated. First, consider a truth NAE assignment $f$ to $\phi$. We construct $T_1,T_2\subseteq E(G)$ and prove that they form an \localtour of $G$. Start by putting $P_i$ in $T_1$ and $\overline{P}_i$ in $T_2$ if $x_i$ is true, and the other way around if $x_i$ is false. From now on, whenever we add a forced edge to $T_1$ and $T_2$, we are actually adding the trails depicted in Figure~\ref{fig:forcedEdgesTrails}. 

\begin{figure}[h]
     \centering
     \begin{subfigure}[b]{0.4\textwidth}
         \centering
         \input{figs/forcedEdgesTrail1}
         \caption{Trail added to $T_1$ when forced edge $uv$ is added to $T_1$.}
     \end{subfigure}
     \hfill
     \begin{subfigure}[b]{0.4\textwidth}
         \centering
         \input{figs/forcedEdgesTrail2}
         \caption{Trail added to $T_2$ when forced edge $uv$ is added to $T_2$.}
     \end{subfigure}
     \caption{Trails related to forced edges.}
     \label{fig:forcedEdgesTrails}
\end{figure}
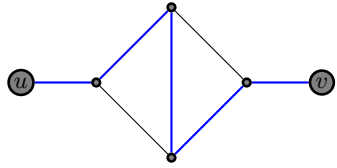
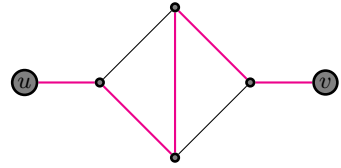

Now, add $E'$ to both $T_1$ and $T_2$, and consider $c_i$ with literals $\ell_1,\ell_2,\ell_3$. Suppose, without loss of generality, that $\ell_1$ is true and $\ell_2$ is false. We then add to $T_1$ the trail depicted in Figure~\ref{fig:l1true}, and to $T_2$ the one depicted in Figure~\ref{fig:l2true}. Observe that all internal edges of $C_i$ are covered; also note that the edges incident to $I_i(\ell_3)$ that come from the variable must be either in $T_1$ or in $T_2$ by construction. We know that the remaining edges are also covered by $T_1\cup T_2$ by construction. Finally, notice that both $T_1$ and $T_2$ touch all odd-degree vertices in a way that every vertex (including the even-degree ones) has even degree in $T_1$ and in $T_2$, except $s_1,s_2$ which have degree exactly~1. Also note that they form a connected graph; indeed they are formed by the cycle passing through the variable gadgets and $t$, together with some pending trails passing by the clause gadgets. Therefore, we can find an $s_1,s_2$-trail passing by all edges of $T_1$, and an $s_2,s_1$-trail passing by all edges of $T_2$, thus getting our \localtour.

\input{figs/DynTrails}

Now, let $T$ be an \localtour of $(G,[2])$, and for each $i\in [2]$, denote by $T_i$ the trail in $G$ defined by $T$ restricted to $G_i$. First observe that Lemma~\ref{lem:oddVertices} indeed ensures that $T_1$ and $T_2$ restricted to the gadget related to a forced edge $uv$ must be exactly as the trails depicted in Figure~\ref{fig:forcedEdgesTrails}; so in what follows we treat them exactly like edges that must appear in $T_1$ and $T_2$. Since there are 2 vertices of degree~1, namely $s_1$ and $s_2$, by Lemma~\ref{lem:oddVertices} we can suppose that $T_1$ starts in $s_1$ and finishes in $s_2$, while $T_2$ starts in $s_2$ and finishes in $s_1$. Therefore, each of $T_1$ and $T_2$ contains a tour in $G-\{s_1,s_2\}$, and hence: 
\begin{itemize}
    \item[(I)] Every vertex $u\in V(G)\setminus \{s_1,s_2\}$ has even degree in both $T_1$ and $T_2$.
\end{itemize}

Now, we prove that for each $x_i$, if $T_1$ intersects $P_i$, then $T_1$ does not intersect $\overline{P}_i$, the same holding for $T_2$. For this, consider $c_i$ with literals $\ell_1$, $\ell_2$, $\ell_3$, and for each $j\in [3]$, denote by $e^1_i(\ell_j),e^2_i(\ell_j)$ the edges incident to $I_i(\ell_j)$ not contained in $C_i$. We first prove that, for each $j\in [3]$ and $k\in [2]$:

\begin{itemize}
    \item[(II)] Edge $e^1_i(\ell_j)$ is used in $T_k$ if and only if edge $e^2_i(\ell_j)$ is used in $T_k$.  
\end{itemize}

Without loss of generality, assume $j=1$ and let $k\in [2]$; suppose by contradiction that $e^1_i(\ell_1)$ is in $T_k$, while $e^2_i(\ell_1)$ is not in $T_k$. By (I) we get that either $a_1^iI_i(\ell_1)$ or $a_2^iI_i(\ell_1)$ is in $T_k$, say $a_1^iI_i(\ell_1)$. But then we get that $a_1^ia_2^i$ and $a^i_2I_i(\ell_1)$ are not in $T_k$, as otherwise either $a_1^i$ or $I_i(\ell_1)$ would have odd degree in $T_k$. This is a contradiction since we then get $a^i_2$ with degree~1 in $T_k$. The same argument can be analogously applied when $a_2^iI_i(\ell_1)$ is in $T_k$ or when $j\in \{2,3\}$.

Consider now a variable $x_i$, and let $c_{i_1},\cdots,c_{i_p}$ be the clauses containing $x_i$ positively and $c_{j_1},\cdots,c_{j_q}$ containing $x_i$ negatively. Because $I_i$ has degree~3 in $G$ and $E'\setminus \{s_it\mid i\in [2]\}$ is contained in $T_1$, we get that exactly one between $I_1I_{i_1}(x_i)$ and $I_1I_{j_1}(\overline{x}_i)$ is contained in $T_1$. From (II), we then get that either $P_i$ is contained in $T_1$ or $\overline{P}_i$ is contained in $T_1$. Observe that this implies that $P_i$ is contained in $T_k$, while $\overline{P}_i$ is contained in $T_{3-k}$, for some $k\in[2]$. We then set $x_i$ to be true if and only if $T_1$ contains $P_i$. Because  the edges in $E_i = \{e^1_i(\ell_j),e^2_i(\ell_j)\mid j\in[3]\}$ separate $C_i$ from the rest of the graph and by (I), we get that both $T_1$ and $T_2$ must intersect $E_i$. Finally by (II) we get that this assignment is a NAE truth assignment for $\phi$.
\end{proof}

Observe that if we add two new vertices of degree one adjacent to vertex $t$, then we get a reduction to the problem of deciding whether the edges of $G$ can be covered by two trails, proving Corollary~\ref{cor:cover}.

\begin{corollary}
\plocaltour and \plocaltrail are $\NP$-complete on temporal graphs with lifetime $\tau$ for every fixed $\tau\ge 2$. This also holds on dynamic-based graphs.
\label{cor:local}
\end{corollary}
\begin{proof}
We first make a reduction from \plocaltour on $(G,[2])$ to \plocaltrail on $(G',[\tau])$. 
Given $(G,[2])$, let $G'$ be obtained from $G$ by adding a star on $\tau+1$ vertices and identifying one of its leaves with a vertex $s\in V(G)$. We argue that $(G,[2])$ has an \localtour starting and finishing in $s$ if and only if $(G',[\tau])$ has an \localtrail. The lemma follows because we can then obtain a Turing reduction by building a distinct instance for each $s\in V(G)$. Denote the vertices of the initial star by $u,v_1,\cdots,v_{\tau+1}$, where $u$ is the central vertex, and $v_2$ is the vertex where $G$ is pending. Let $T$ be an \localtour of $(G,[2])$ starting and finishing in $s$, and $T_1,T_2$ be the trails in $G$ defined by $T$.  Build an \localtrail of $(G',[\tau])$ by visiting $v_1u$, $v_2u$ and $T_1$ in $G'_1$, then performing $T_2$ and visiting $v_2u$ and $uv_3$ in $G'_2$, and finish to visit the remaining edges of the star in the obvious way.

Now, let $T$ be an \localtrail of $(G',[\tau])$, and denote by $T_i$ the trail in $G'$ defined by $T$ restricted to $G'_i$, for each $i\in [\tau]$. Observe that because we have $\tau+1$ cut edges, we get that each $T_i$ contains at most~2 of them, and in case it contains exactly~2, say $v_1u,uv_2$, then $T_{i+1}$ either does not contain any cut edge, or must intersect $T_i$ in $v_1u,uv_2$. This means that the best we can do in order to finish by time $\tau$ is to visit exactly two of them in the first snapshot, and exactly one more in each of the subsequent snapshots. We can therefore suppose, without loss of generality that $T_i$ contains $v_iu,v_{i+1}u$ for each $i\in [\tau]$. Note that this implies that every edge of $(G,[2])$ must be visited in $T_1$ and $T_2$, with $T_1$ starting in $v_2$ and $T_2$ finishing in $v_2$, as we wanted to prove. 

Finally, note that $(G',[\tau])$ constructed above has an \localtrail if and only if $(G',\tau+1)$ has an \localtour. This completes our proof.
\end{proof}

\section{Eulerian Tours and Trails}

We finally focus on \ptrail and \ptour, proving that in the general case they are both \NP-complete, hence, proving Item~\ref{item:three} in Theorem~\ref{thm:main}. To this aim, we make an adaptation of the construction in Theorem~\ref{thm:localtour}. Observe that here the base graph needs to be Eulerian as otherwise the answer to \ptrail is trivially NO. This also implies that the problem restricted to dynamic-based graphs is trivial: if the base graph is Eulerian, then the answer to \ptour is YES; otherwise, then the answer is NO. The trick now is to take advantage of the function $\lambda$ in order to enforce the edges. 

\begin{theorem}
\ptour and \ptrail are $\NP$-complete, even on temporal graphs with fixed lifetime~$\tau\ge 2$.
\end{theorem}
\begin{proof}
We first prove the case $\tau=2$. For this, we simply replace the gadget to enforce an edge $uv$ in the construction of Section~\ref{sec:localtrail} by two paths of length~2, $P^1_{uv}$ and $P^2_{uv}$, where the edges in $P^i_{uv}$ are active only in snapshot $G_i$, for each $i\in [2]$. Because the arguments used in Section~\ref{sec:localtrail} depended only on the fact of $uv$ be indeed an enforced edge, we can apply the same arguments here. The only difference is that the trails in $G_1$ and $G_2$ now cannot intersect, which indeed is the case since the intersection between $T_1$ and $T_2$ in Section~\ref{sec:localtrail} is exactly the set of forced edges, and since here each appearance of a forced edge $uv$ is actually related either to $P^1_{uv}$ or to $P^2_{uv}$.

Now, in order to prove the $\NP$-completeness for higher values of $\tau$, we can simply add new vertices $v_3,\cdots,v_{\tau}$ and edges $\{s_1v_3\}\cup \{v_iv_{i+1}\mid i\in \{3,\cdots,\tau\}\}$. This gives us that \ptrail is $\NP$-complete on $(G,\lambda)$ with lifetime $\tau$ for every fixed $\tau\ge 2$. And if we want a closed trail, it suffices to identify $v_\tau$ with $s_1$, if $\tau\ge 4$, and if $\tau=3$, we add a new vertices $v_4$ and edges $v_3v_4,v_4s_1$ active in snapshot $G_3$. This concludes our proof.
\end{proof}


\bibliographystyle{plain}
\bibliography{references.bib}

\appendix

\section{A brief summary about Eulerian tours in dynamic graphs}
\label{app:rel}
A \emph{dynamic graph} is a pair $(G,T)$ where $G$ is a finite digraph and $T$ is a function $T:E(G)\rightarrow \Z$, called \emph{transit time function}. A dynamic graph can also be seen as a special type of infinite digraph ${\cal G}$, where $V({\cal G}) = V(G)\times \Z$, and $(u,i)(v,j)\in E(\cal{G})$ if and only if $uv\in E(G)$ and $T(uv) = j-i$. Observe that the transit time of an arc can also be negative, and therefore there might exist arcs going from a vertex $(u,i)$ to a vertex $(v,j)$ with $j<i$, which in the temporal graph context would be considered as going back in time. 
An \emph{Eulerian trail} in $(G,T)$ is a trail that passes through all the edges of ${\cal G}$. More formally, it is a function $f:\Z\rightarrow V(G)\times \Z$ such that $f(i) f(i+1)\in E({\cal G})$ for every $i\in \Z$, and for every $(u,i)(v,j)\in E({\cal G})$, there exists a unique $\ell$ such that $(u,i)(v,j)$ is equal to $f(\ell)f({\ell+1})$. 

Recall that a digraph $G$ is Eulerian if and only: (i) $G$ has at most one non-trivial component; and (ii) the indegree of $u$, denoted with $d^-(u)$, is equal to its outdegree, denoted as $d^+(u)$, for every $u\in V(G)$. Observe that these conditions are also trivially necessary for the infinite case. Also, note that, given $u\in V(G)$, each in-arc $uv$ of $G$ incident to $u$ gives rise to exactly one in-arc $(u,i)(v,i+T(u,v))$ incident to $(u,i)$, for every $i\in \Z$. The same clearly holds for every out-arc. Therefore, one can see that ${\cal G}$ satisfies (i) and (ii) if and only if  $G$ satisfies (i) and (ii). However, as proved in~\cite{orlin1984some}, these are not the only necessary conditions. Nevertheless, a characterization is still possible, with an additional, also easy to test, condition.

\begin{theorem}[\cite{orlin1984some}]
Let $(G,T)$ be a dynamic graph. Then $(G,T)$ has an Eulerian trail if and only the following conditions hold:
\begin{enumerate}
    \item $d^-(u) = d^+(u)$ for every $u\in V(G)$;
    \item $G$ is connected; and
    \item $\sum_{e\in E(G)}T(e) \in \{-1,1\}$
\end{enumerate}
\end{theorem}

Even if the necessary part of the proof is more technical, the sufficiency part is quite natural, because condition (3) tells us that, given an eulerian tour $T= (v_1,\cdots,v_{m},v_1)$ of $G$ and fixing a time $i$, we can use $T$ to traverse all the edges incident to $(v_1,i)$ in a way that we arrive in $(v_1,i+1)$ (or $(v_1,i-1)$ if the sum is $-1$) just in time to apply the same process to $(v_1,i+1)$. Because $(v_1,i)$ is chosen arbitrarily, we are ensured to visit all edges of ${\cal G}$.

\section{Proofs}

\subsection{Proof of Lemma~\ref{lem:easy}}
\label{sec:easy}
\begin{proof}
Let $G_1,\cdots, G_\tau$ be the snapshots of $G$; note first that if $E(G_i)$ is empty, then this snapshot can be suppressed. Our problem reduces to deciding whether there is a choice of connected components $H_1, \ldots, H_\tau$, one for each timestamp $i$, that together cover all the edges of $G$ and is such that $H_i$ intersects $H_{i+1}$, for each $i\in [\tau-1]$. As for each $i\in[\tau]$, there are at most $n$ nodes in the intersections, there are at most $O(n^{\tau-1})$ choices. For each choice the test can be done in $O(\tau(n+m))$, obtaining $O(\tau(n+m)n^{\tau-1})$ running time, which is $O((n+m)n^{\tau-1})$.
\end{proof}

\subsection{Proof of Lemma~\ref{lem:oddVertices}}
\label{app:oddVertices}
\begin{proof}
For each $i\in [2]$, denote by $T_i$ the trail in $G$ equal to $T$ restricted to timestamp $i$, and suppose, by contradiction, that $u\in V(G)$ is a vertex with odd degree not contained in $T_1$. Because $T$ is a temporal tour, observe that $T_1$ is a trail in $G$ starting at some $s$ and finishing at some $t$, and $T_2$ is a trail in $G$ starting at $t$ and finishing at $s$, with possibly $s=t$. This means that the subgraph of $G$ formed by the edges of $T_2$ is such that every $x\in V(G)\setminus \{s,t\}$ has even degree. This is a contradiction because, since no edge incident to $u$ is visited in $T_1$, we get that all the edges incident to $u$ must be visited in $T_2$, i.e., $u$ would have odd degree in $T_2$. The same argument holds in case $u$ is not in $T_2$, and the lemma follows.
\end{proof}

\end{document}

%% file: figs/walkGadget.tex
\begin{center}
  \begin{tikzpicture}[scale=1]
  \pgfsetlinewidth{1pt}
  \pgfdeclarelayer{bg}    
   \pgfsetlayers{bg,main}  

  \tikzset{vertex/.style={circle, minimum size=0.1cm, draw, inner sep=1pt, fill=black!20, ,thin}}
  \tikzset{snapshot/.style ={draw=black!50, rounded corners, dashed, minimum height=8mm, minimum width=2.3cm} }
  \tikzset{subgraph/.style ={draw=black!50, circle, draw, minimum width=1cm, yscale=2, fill=white} }

    \node [label=-90:$G_1$, draw=black!50, rounded corners, dashed, minimum height=4cm, minimum width=3.8cm] at (0,0){};
    \node [label=-90:$G_2$, draw=black!50, rounded corners, dashed, minimum height=4cm, minimum width=3.8cm] at (4,0){};
    \node [label=-90:$G_3$, draw=black!50, rounded corners, dashed, minimum height=4cm, minimum width=3.8cm] at (8,0){};

    \node[vertex] (x1) at (-1,0.5) {$x_1$};
    \node[vertex] (a1) at (-1.5,-0.5) {$a_1$};
    \node[vertex] (a2) at (-0.5,-0.5) {$a_2$};
    \node[vertex] (b1) at (-1.5,-1.5) {$b_1$};
    \node[vertex] (b2) at (-0.5,-1.5) {$b_2$};
    
    \node[vertex] (nx1) at (1,0.5) {$\overline{x}_1$};
    \node[vertex] (a3) at (0.5,-0.5) {$a_3$};
    \node[vertex] (a4) at (1.5,-0.5) {$a_4$};
    \node[vertex] (b3) at (0.5,-1.5) {$b_3$};
    \node[vertex] (b4) at (1.5,-1.5) {$b_4$};
    \begin{pgfonlayer}{bg}    
        \draw (x1)--(a1)--(b1) (x1)--(a2)--(b2) (nx1)--(a3)--(b3) (nx1)--(a4)--(b4);
    \end{pgfonlayer}
    
    \node[vertex, shift={(4,0)}] (x1) at (-1.5,0.5) {$x_1$};
    \node[vertex, shift={(4,0)}] (nx1) at (-0.5,0.5) {$\overline{x}_1$};
    \node[vertex, shift={(4,0)}] (x2) at (0.5,0.5) {$x_2$};
    \node[vertex, shift={(4,0)}] (nx2) at (1.5,0.5) {$\overline{x}_2$};
    \node[vertex, shift={(4,0)}] (T) at (0,1.5) {$T$};    
    \begin{pgfonlayer}{bg}    
        \draw (x1)--(T)--(nx1) (x2)--(T)--(nx2);
    \end{pgfonlayer}
  
    \node[vertex, shift={(8,0)}] (x1) at (-1,0.5) {$x_2$};
    \node[vertex, shift={(8,0)}] (a1) at (-1.5,-0.5) {$a_1$};
    \node[vertex, shift={(8,0)}] (a2) at (-0.5,-0.5) {$a_2$};
    \node[vertex, shift={(8,0)}] (b1) at (-1.5,-1.5) {$b_1$};
    \node[vertex, shift={(8,0)}] (b2) at (-0.5,-1.5) {$b_2$};
    
    \node[vertex, shift={(8,0)}] (nx1) at (1,0.5) {$\overline{x}_2$};
    \node[vertex, shift={(8,0)}] (a3) at (0.5,-0.5) {$a_3$};
    \node[vertex, shift={(8,0)}] (a4) at (1.5,-0.5) {$a_4$};
    \node[vertex, shift={(8,0)}] (b3) at (0.5,-1.5) {$b_3$};
    \node[vertex, shift={(8,0)}] (b4) at (1.5,-1.5) {$b_4$};
    \begin{pgfonlayer}{bg}    
        \draw (x1)--(a1)--(b1) (x1)--(a3)--(b3) (nx1)--(a2)--(b2) (nx1)--(a4)--(b4);
    \end{pgfonlayer}

  \end{tikzpicture}

\end{center}

%% file: figs/outerplanar.tex
\begin{center}
  \begin{tikzpicture}[scale=0.5]
  \pgfsetlinewidth{1pt}
  \pgfdeclarelayer{bg}    
   \pgfsetlayers{bg,main}  

  \tikzset{vertex/.style={circle, minimum size=0.1cm, draw, inner sep=1pt, fill=black!50}}

    \node[vertex] (a) at (0:1) {};
    \node[vertex] (b) at (60:1) {};
    \node[vertex] (c) at (120:1) {};
    \node[vertex] (d) at (180:1) {};
    \node[vertex] (e) at (240:1) {};
    \node[vertex] (f) at (300:1) {};

    \begin{pgfonlayer}{bg}    
        \draw (a)--(b)--(c)--(d)--(e)--(f)--(a);
    \end{pgfonlayer}
    
    \node[vertex,xshift=0.75cm,yshift=0.433cm] (a) at (0:1) {};
    \node[vertex,xshift=0.75cm,yshift=0.433cm] (b) at (60:1) {};
    \node[vertex,xshift=0.75cm,yshift=0.433cm] (c) at (120:1) {};
    \node[vertex,xshift=0.75cm,yshift=0.433cm] (d) at (180:1) {};
    \node[vertex,xshift=0.75cm,yshift=0.433cm] (e) at (240:1) {};
    \node[vertex,xshift=0.75cm,yshift=0.433cm] (f) at (300:1) {};

    \begin{pgfonlayer}{bg}    
        \draw (a)--(b)--(c)--(d)--(e)--(f)--(a);
    \end{pgfonlayer}

    \node[vertex,xshift=0.75cm,yshift=1.3cm] (a) at (0:1) {};
    \node[vertex,xshift=0.75cm,yshift=1.3cm] (b) at (60:1) {};
    \node[vertex,xshift=0.75cm,yshift=1.3cm] (c) at (120:1) {};
    \node[vertex,xshift=0.75cm,yshift=1.3cm] (d) at (180:1) {};
    \node[vertex,xshift=0.75cm,yshift=1.3cm] (e) at (240:1) {};
    \node[vertex,xshift=0.75cm,yshift=1.3cm] (f) at (300:1) {};

    \begin{pgfonlayer}{bg}    
        \draw (a)--(b)--(c)--(d)--(e)--(f)--(a);
    \end{pgfonlayer}

    \node[vertex,xshift=1.5cm] (a) at (0:1) {};
    \node[vertex,xshift=1.5cm] (b) at (60:1) {};
    \node[vertex,xshift=1.5cm] (c) at (120:1) {};
    \node[vertex,xshift=1.5cm] (d) at (180:1) {};
    \node[vertex,xshift=1.5cm] (e) at (240:1) {};
    \node[vertex,xshift=1.5cm] (f) at (300:1) {};

    \begin{pgfonlayer}{bg}    
        \draw (a)--(b)--(c)--(d)--(e)--(f)--(a);
    \end{pgfonlayer}
    
    \node[vertex,yshift=-0.866cm] (a) at (0:1) {};
    \node[vertex,yshift=-0.866cm] (b) at (60:1) {};
    \node[vertex,yshift=-0.866cm] (c) at (120:1) {};
    \node[vertex,yshift=-0.866cm] (d) at (180:1) {};
    \node[vertex,yshift=-0.866cm] (e) at (240:1) {};
    \node[vertex,yshift=-0.866cm] (f) at (300:1) {};

    \begin{pgfonlayer}{bg}    
        \draw (a)--(b)--(c)--(d)--(e)--(f)--(a);
    \end{pgfonlayer}

    \node[vertex,xshift=0.75cm,yshift=-1.3cm] (a) at (0:1) {};
    \node[vertex,xshift=0.75cm,yshift=-1.3cm] (b) at (60:1) {};
    \node[vertex,xshift=0.75cm,yshift=-1.3cm] (c) at (120:1) {};
    \node[vertex,xshift=0.75cm,yshift=-1.3cm] (d) at (180:1) {};
    \node[vertex,xshift=0.75cm,yshift=-1.3cm] (e) at (240:1) {};
    \node[vertex,xshift=0.75cm,yshift=-1.3cm] (f) at (300:1) {};

    \begin{pgfonlayer}{bg}    
        \draw (a)--(b)--(c)--(d)--(e)--(f)--(a);
    \end{pgfonlayer}

    \node[vertex,xshift=-0.75cm,yshift=-1.3cm] (a) at (0:1) {};
    \node[vertex,xshift=-0.75cm,yshift=-1.3cm] (b) at (60:1) {};
    \node[vertex,xshift=-0.75cm,yshift=-1.3cm] (c) at (120:1) {};
    \node[vertex,xshift=-0.75cm,yshift=-1.3cm] (d) at (180:1) {};
    \node[vertex,xshift=-0.75cm,yshift=-1.3cm] (e) at (240:1) {};
    \node[vertex,xshift=-0.75cm,yshift=-1.3cm] (f) at (300:1) {};

    \begin{pgfonlayer}{bg}    
        \draw (a)--(b)--(c)--(d)--(e)--(f)--(a);
    \end{pgfonlayer}

    \node[vertex,xshift=-0.75cm,yshift=0.433cm] (a) at (0:1) {};
    \node[vertex,xshift=-0.75cm,yshift=0.433cm] (b) at (60:1) {};
    \node[vertex,xshift=-0.75cm,yshift=0.433cm] (c) at (120:1) {};
    \node[vertex,xshift=-0.75cm,yshift=0.433cm] (d) at (180:1) {};
    \node[vertex,xshift=-0.75cm,yshift=0.433cm] (e) at (240:1) {};
    \node[vertex,xshift=-0.75cm,yshift=0.433cm] (f) at (300:1) {};

    \begin{pgfonlayer}{bg}    
        \draw (a)--(b)--(c)--(d)--(e)--(f)--(a);
    \end{pgfonlayer}
    
    \node[vertex,xshift=-0.75cm,yshift=1.3cm] (a) at (0:1) {};
    \node[vertex,xshift=-0.75cm,yshift=1.3cm] (b) at (60:1) {};
    \node[vertex,xshift=-0.75cm,yshift=1.3cm] (c) at (120:1) {};
    \node[vertex,xshift=-0.75cm,yshift=1.3cm] (d) at (180:1) {};
    \node[vertex,xshift=-0.75cm,yshift=1.3cm] (e) at (240:1) {};
    \node[vertex,xshift=-0.75cm,yshift=1.3cm] (f) at (300:1) {};

    \begin{pgfonlayer}{bg}    
        \draw (a)--(b)--(c)--(d)--(e)--(f)--(a);
    \end{pgfonlayer}

    \node[vertex,xshift=-1.5cm] (a) at (0:1) {};
    \node[vertex,xshift=-1.5cm] (b) at (60:1) {};
    \node[vertex,xshift=-1.5cm] (c) at (120:1) {};
    \node[vertex,xshift=-1.5cm] (d) at (180:1) {};
    \node[vertex,xshift=-1.5cm] (e) at (240:1) {};
    \node[vertex,xshift=-1.5cm] (f) at (300:1) {};
    
    \begin{pgfonlayer}{bg}    
        \draw (a)--(b)--(c)--(d)--(e)--(f)--(a);
    \end{pgfonlayer}

\end{tikzpicture}

\end{center}

%% file: figs/DynEdge.tex
\begin{center}
  \begin{tikzpicture}[scale=1]
  \pgfsetlinewidth{1pt}
  \pgfdeclarelayer{bg}    
   \pgfsetlayers{bg,main}  

  \tikzset{vertex/.style={circle, minimum size=0.1cm, draw, inner sep=1pt, fill=black!50}}
  \tikzset{blackbox/.style ={draw=black!50, rounded corners, dashed}}
    
    \node[blackbox] (Ci1) at (2,1) {$C_{i_1}$};
    \node[blackbox] (Ci2) at (4,1) {$C_{i_2}$};
    \node at (5,1) {$\cdots$};
    \node[blackbox] (Cip) at (7,1) {$C_{i_p}$};
    \node[blackbox] (Cj1) at (2,-1) {$C_{j_1}$};
    \node[blackbox] (Cj2) at (4,-1) {$C_{j_2}$};
    \node at (5,-1) {$\cdots$};
    \node[blackbox] (Cjq) at (7,-1) {$C_{j_q}$};
    \node[vertex] (I) at (0,0) {$I_i$};
    \node[vertex] (O) at (9,0) {$O_i$};
    
    \begin{pgfonlayer}{bg}    
        \draw (I)--(Ci1)--(Ci2)--(4.5,1) (5.5,1)--(Cip)--(O);
        \draw (I)--(Cj1)--(Cj2)--(4.5,-1) (5.5,-1)--(Cjq)--(O);
    \end{pgfonlayer}
    
  \end{tikzpicture}

\end{center}

%% file: figs/forcedEdge.tex
\begin{center}
  \begin{tikzpicture}[scale=1]
  \pgfsetlinewidth{1pt}
  \pgfdeclarelayer{bg}    
   \pgfsetlayers{bg,main}  

  \tikzset{vertex/.style={circle, minimum size=0.1cm, draw, inner sep=1pt, fill=black!50}}

    \node[vertex] (a1) at (0,0) {$u$};
    \node[vertex] (a2) at (1,0) {};
    \node[vertex] (a3) at (2,1) {};
    \node[vertex] (a4) at (2,-1) {};
    \node[vertex] (a5) at (3,0) {};
    \node[vertex] (a6) at (4,0) {$v$};

    \begin{pgfonlayer}{bg}    
        \draw (a1)--(a2)--(a3)--(a4)--(a5)--(a6) (a2)--(a4) (a3)--(a5);
    \end{pgfonlayer}
    
  \end{tikzpicture}

\end{center}

%% file: figs/DynClause.tex
\begin{center}
  \begin{tikzpicture}[scale=1]
  \pgfsetlinewidth{1pt}
  \pgfdeclarelayer{bg}    
   \pgfsetlayers{bg,main}  

  \tikzset{vertex/.style={circle, minimum size=0.1cm, draw, inner sep=1pt, fill=black!50}}

    \node[vertex, label=90:$c^i_{2}$] (c2) at (-0.8,1.4) {};
    \node[vertex, label=90:$a^i_1$] (a1) at (0.8,1.4) {};
    \node[vertex, label=0:$a^i_{2}$] (a2) at (2,0) {};
    \node[vertex, label=0:$b^i_1$] (b1) at (0.8,-1.4) {};
    \node[vertex, label=180:$b^i_{2}$] (b2) at (-0.8,-1.4) {};
    \node[vertex, label=180:$c^i_1$] (c1) at (-2,0) {};
    
    \node[vertex, label=0:$I_i(\ell_1)$] (l1) at (2,1.4) {};
    \node[vertex, label=-90:$I_i(\ell_2)$] (l2) at (0,-2) {};
    \node[vertex, label=180:$I_i(\ell_3)$] (l3) at (-2,1.4) {};

    \begin{pgfonlayer}{bg}    
        \draw[red,thick] (a1)--(c2) (a2)--(b1) (c1)--(b2);
        \draw (a1)--(a2) (b1)--(b2) (c1)--(c2) (a1)--(l1)--(a2) (b1)--(l2)--(b2) (c1)--(l3)--(c2);
    \end{pgfonlayer}
    
  \end{tikzpicture}

\end{center}

%% file: figs/forcedEdgesTrail1.tex
\begin{center}
  \begin{tikzpicture}[scale=1]
  \pgfsetlinewidth{1pt}
  \pgfdeclarelayer{bg}    
   \pgfsetlayers{bg,main}  

  \tikzset{vertex/.style={circle, minimum size=0.1cm, draw, inner sep=1pt, fill=black!50}}

    \node[vertex] (a1) at (0,0) {$u$};
    \node[vertex] (a2) at (1,0) {};
    \node[vertex] (a3) at (2,1) {};
    \node[vertex] (a4) at (2,-1) {};
    \node[vertex] (a5) at (3,0) {};
    \node[vertex] (a6) at (4,0) {$v$};

    \begin{pgfonlayer}{bg}    
        \draw[blue,thick] (a1)--(a2)--(a3)--(a4)--(a5)--(a6);
        \draw (a2)--(a4) (a3)--(a5);
    \end{pgfonlayer}
    
  \end{tikzpicture}

\end{center}

%% file: figs/forcedEdgesTrail2.tex
\begin{center}
  \begin{tikzpicture}[scale=1]
  \pgfsetlinewidth{1pt}
  \pgfdeclarelayer{bg}    
   \pgfsetlayers{bg,main}  

  \tikzset{vertex/.style={circle, minimum size=0.1cm, draw, inner sep=1pt, fill=black!50}}

    \node[vertex] (a1) at (0,0) {$u$};
    \node[vertex] (a2) at (1,0) {};
    \node[vertex] (a3) at (2,1) {};
    \node[vertex] (a4) at (2,-1) {};
    \node[vertex] (a5) at (3,0) {};
    \node[vertex] (a6) at (4,0) {$v$};

    \begin{pgfonlayer}{bg}    
        \draw[magenta,thick] (a1)--(a2)--(a4)--(a3)--(a5)--(a6);
        \draw (a2)--(a3) (a4)--(a5);
    \end{pgfonlayer}
    
  \end{tikzpicture}

\end{center}

%% file: figs/DynTrails.tex
\begin{figure}[h]
     \centering
     \begin{subfigure}[h]{0.4\textwidth}
         \centering
\begin{center}
  \begin{tikzpicture}[scale=1]
  \pgfsetlinewidth{1pt}
  \pgfdeclarelayer{bg}    
   \pgfsetlayers{bg,main}  

  \tikzset{vertex/.style={circle, minimum size=0.1cm, draw, inner sep=1pt, fill=black!50}}

    \node[vertex, label=90:$c^i_{2}$] (c2) at (-0.8,1.4) {};
    \node[vertex, label=90:$a^i_1$] (a1) at (0.8,1.4) {};
    \node[vertex, label=0:$a^i_{2}$] (a2) at (2,0) {};
    \node[vertex, label=0:$b^i_1$] (b1) at (0.8,-1.4) {};
    \node[vertex, label=180:$b^i_{2}$] (b2) at (-0.8,-1.4) {};
    \node[vertex, label=180:$c^i_1$] (c1) at (-2,0) {};
    
    \node[vertex, label=90:$I_i(\ell_1)$] (l1) at (2,1.4) {};
    \node[vertex, label=0:$I_i(\ell_2)$] (l2) at (0,-2) {};
    \node[vertex, label=90:$I_i(\ell_3)$] (l3) at (-2,1.4) {};

    \begin{pgfonlayer}{bg}    
        \draw[blue,thick] (a1)--(c2) (a2)--(b1) (c1)--(b2)  (a1)--(l1)--(a2) (b1)--(b2) (3,1)--(l1)--(3,1.4) (c1)--(c2);
        \draw (a1)--(a2)  (b1)--(l2)--(b2) (-3,1)--(l3)--(-3,1.4) (-0.2,-3)--(l2)--(0.2,-3) (c1)--(l3)--(c2);
    \end{pgfonlayer}

  \end{tikzpicture}

\end{center}
         \caption{Added to $T_1$ when $\ell_1$ is true.}
         \label{fig:l1true}
     \end{subfigure}
     \hfill
     \begin{subfigure}[h]{0.4\textwidth}
         \centering
         
         \begin{center}
  \begin{tikzpicture}[scale=1]
  \pgfsetlinewidth{1pt}
  \pgfdeclarelayer{bg}    
   \pgfsetlayers{bg,main}  

  \tikzset{vertex/.style={circle, minimum size=0.1cm, draw, inner sep=1pt, fill=black!50}}

        \node[vertex, label=90:$c^i_{2}$] (c2) at (-0.8,1.4) {};
    \node[vertex, label=90:$a^i_1$] (a1) at (0.8,1.4) {};
    \node[vertex, label=0:$a^i_{2}$] (a2) at (2,0) {};
    \node[vertex, label=0:$b^i_1$] (b1) at (0.8,-1.4) {};
    \node[vertex, label=180:$b^i_{2}$] (b2) at (-0.8,-1.4) {};
    \node[vertex, label=180:$c^i_1$] (c1) at (-2,0) {};
    
    \node[vertex, label=90:$I_i(\ell_1)$] (l1) at (2,1.4) {};
    \node[vertex, label=0:$I_i(\ell_2)$] (l2) at (0,-2) {};
    \node[vertex, label=90:$I_i(\ell_3)$] (l3) at (-2,1.4) {};

    \begin{pgfonlayer}{bg}    
        \draw[magenta,thick] (a1)--(c2) (a2)--(b1) (c1)--(b2)     (-0.2,-3)--(l2)--(0.2,-3) (b1)--(l2)--(b2) (c1)--(l3)--(c2) (a1)--(a2);
        \draw  (-3,1)--(l3)--(-3,1.4) (b1)--(b2) (3,1)--(l1)--(3,1.4) (c1)--(c2) (a1)--(l1)--(a2) ;
    \end{pgfonlayer}

  \end{tikzpicture}

         \end{center}
         \caption{Added to $T_2$ when $\ell_2$ is false.}
         \label{fig:l2true}
     \end{subfigure}
    \caption{Trails in $C_i$ related to a given NAE assignment.}  
    \label{fig:DynTrails}
\end{figure}
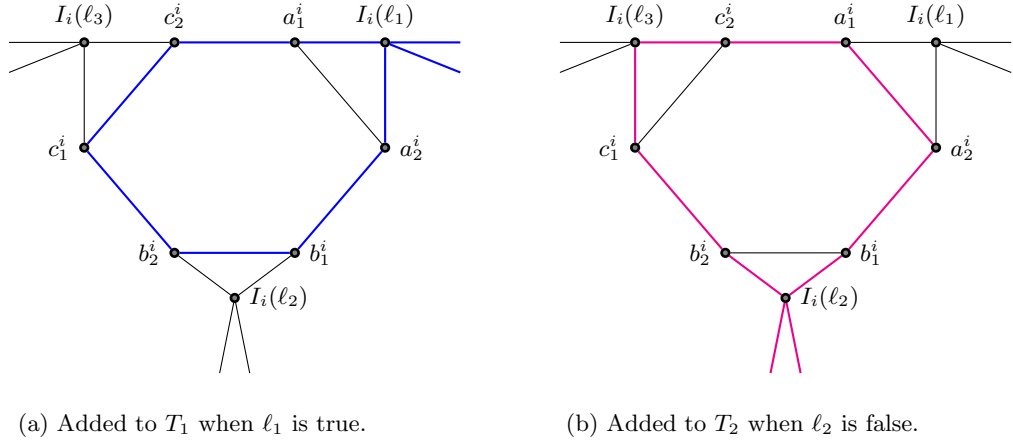